\documentclass{article}

\usepackage[margin=1.125in]{geometry}
\usepackage{xcolor}
\definecolor{DarkGreen}{rgb}{0.1,0.5,0.1}
\definecolor{DarkRed}{rgb}{0.5,0.1,0.1}
\definecolor{DarkBlue}{rgb}{0.1,0.1,0.5}

\usepackage[small]{caption}
\usepackage{amssymb,amsmath,amsthm,amsfonts,enumitem}
\usepackage{fullpage,nicefrac,comment}
\usepackage{tikz}
\usetikzlibrary{arrows,decorations.pathmorphing,decorations.shapes,snakes,patterns,positioning}
\usepackage[ruled,linesnumbered]{algorithm2e}

\newcommand{\st}[1]{{\mathrm{\textbf{set}}}\left(#1\right)}

\newcommand{\cC}{\ensuremath{\mathcal{C}}}

\newcommand{\F}{{F}}

\newcommand{\ip}[2]{\ensuremath{\left\langle #1,#2\right\rangle}}

\newcommand{\inset}[1]{\left\{#1\right\}}

\newcommand{\inparen}[1]{\left(#1\right)}

\newcommand{\suchthat}{\,:\,}

\newcommand{\spn}{\ensuremath{\operatorname{span}}}

\newcommand{\tr}{\mathrm{tr}}

\newcommand{\Img}{\ensuremath{\operatorname{Im}}}

\newcommand{\eps}{\varepsilon}
\renewcommand{\epsilon}{\varepsilon}
\newcommand{\vphi}{\varphi}

\newtheorem{theorem}{Theorem} 
\newtheorem{lemma}{Lemma} 
\newtheorem{definition}{Definition}

\newtheorem{observation}{Observation}

\newtheorem{remark}{Remark}
\newtheorem{claim}{Claim}
\newtheorem{proposition}{Proposition}

\title{Repairing multiple failures for scalar MDS codes}
\author{
Jay Mardia\thanks{Department of Electrical Engineering, Stanford University. \texttt{jmardia@stanford.edu}}
\ and
Burak Bartan\thanks{Department of Electrical Engineering, Stanford University.  \texttt{bbartan@stanford.edu}} 
\ and Mary Wootters\thanks{Departments of Computer Science and Electrical Engineering, Stanford University.  \texttt{marykw@stanford.edu}}}
\begin{document}
\maketitle

\begin{abstract}
In distributed storage, erasure codes---like Reed-Solomon Codes---are often employed to provide reliability.
In this setting, it is desirable to be able to repair one or more failed nodes while minimizing the \emph{repair bandwidth}.
In this work, motivated by Reed-Solomon codes, we study the problem of repairing multiple failed nodes in a scalar MDS code.
We extend the framework of (Guruswami and Wootters, 2017) to give a framework for constructing repair schemes for multiple failures in general scalar MDS codes, in the centralized repair model.  We then specialize our framework to Reed-Solomon codes, and extend and improve upon recent results of (Dau et al., 2017).  
\end{abstract}

\section{Introduction}\label{sec:intro}
In coding for distributed storage, one wishes to store some data $x \in \Sigma^k$ across $n$ nodes.  These nodes will occasionally fail, and \em erasure coding \em is used to allow for the recovery of $x$ given only a subset of the $n$ nodes.  A common solution is to use a \em Maximum-Distance Separable \em (MDS) code; for example, a Reed-Solomon code.   An MDS code encodes a message $x \in \Sigma^k$ into $n$ symbols $c \in \Sigma^n$, in such a way that any $k$ symbols of $c$ determine $x$.  By putting the symbols $c_i$ of $c$ on different nodes, this gives a distributed storage scheme which can tolerate $n-k$ node failures.

While this level of worst-case robustness is desirable, in practice it is much more common for only a few nodes to fail, rather than $n-k$ of them. 
To that end, it is desirable to design codes which are simultaneously MDS and which also admit cheap repair of a few failures.  One important notion of ``cheap" is \em network bandwidth: \em the amount of data downloaded from the surviving nodes.  The naive MDS repair scheme would involve downloading $k$ complete symbols of $n$.  
Minimum storage regenerating (MSR) codes~\cite{DGWW10} improve the situation; these are codes which maintain the MDS property, while substantially reducing repair bandwidth for a single failure. 

Most of the work in regenerating codes has focused on this case of a single failure, as is many systems this is the most common case, as documented by Rashmi et al. in~\cite{facebook}.  However, even in~\cite{facebook} it is not uncommon to have multiple failures at once, and some systems employ \em lazy repair \em to encourage this~\cite{totalrecall}.  
Motivated by this, many recent works have considered this case of \em multiple failures. \em 
In this work, we focus on the question of multiple failures for scalar MDS codes.
Our work is inspired by Reed-Solomon codes---arguably the most commonly-used code for distributed storage---but our framework works more broadly for any scalar MDS code.

\subsection{Previous work and our contributions}
There has been a huge amount of work on regenerating codes, and we refer the reader to the survey~\cite{survey} for an excellent introduction.  Most of the work has focused on a single failures, but recently there has been a great deal of work on multiple failures.  Two commonly studied models are the \em centralized model \em (which we study here), and the \em cooperative model. \em  In the centralized model, a single repair center is responsible for the repair of all failed nodes, while in the cooperative model the replacement nodes may cooperate but are distinct~\cite{SH13, KLS11, LL14}. 

We focus on the centralized model.  Most of this work in this model has focused on achieving the \em cut-set bound \em for multiple failures~\cite{CJM13,LLL15,RKV16,ZW16,ZW17b,ZW17,WTB17,YB17,YB18}. This extends with well-known cut-set bound for the single-failure case~\cite{DGWW10}, and is only achievable when the \em sub-packetization \em (that is, the number of sub-symbols that each node stores) is reasonably large; in particular, we (at least) require the subpacketization $t$ to be larger than $n-k$, otherwise the trivial lower bound of $k + t - 1$ is larger than the cut-set bound.  Most of the works mentioned above focus on \em array codes, \em that is, codes where the alphabet $\Sigma = B^t$ is naturally thought of as a vector space over a finite field $B$, and the codes are generally $B$-linear.

Other recent works~\cite{multfail,multfail2,YB17} focused on Reed-Solomon codes, and studied multiple failures for \em scalar codes, \em where the alphabet $\Sigma = F$ is a finite field, and the codes are required to be linear over $F$.
In \cite{YB17}, the goal is again the cut-set bound, and the underlying subpacketization is necessarily exponentially large.
In \cite{multfail, multfail2}, 
the sub-packetization is taken to be smaller, on the order of $\log(n)$.  This is the natural parameter regime for Reed-Solomon codes, and in this regime the cut-set bound is not achievable for high-rate codes.  Our work falls into this latter category.

Beginning with the work of Shanmugam et al. in \cite{scalarmds}, the repair properties of scalar MDS codes has been increasingly studied~\cite{scalarmds, gw16, YB16, YTB17, DM17, multfail, multfail2, YB17, DD17, DDC18}.  
In \cite{gw16}, the authors gave a framework for studying single-failure repair schemes for scalar MDS codes, and
the works of Dau et al. \cite{multfail, multfail2} mentioned above
adapt the single-failure scheme from~\cite{gw16} to handle two or three failures for Reed-Solomon codes, in several models, including the centralized model.  

In this work, we extend and improve the results of \cite{multfail,multfail2}.
More precisely, we make the following contributions.
\begin{enumerate}
	\item Following the setup of~\cite{gw16}, we give a general framework for constructing repair schemes of scalar MDS codes for multiple failures.  Theorem~\ref{thm:main} shows that collections of dual codewords with certain properties naturally give rise to repair schemes for multiple failures.  This framework is applicable to any scalar MDS code, and for any number of failures $r \leq n-k$.
	\item We instantiate Theorem~\ref{thm:main} with two different schemes in Theorems~\ref{thm:mainrs_large} and \ref{thm:mainrs_small} that both give non-trivial repair schemes for multiple failures. While the scheme in Theorem~\ref{thm:mainrs_large} is asympototically better (as $n$, the length of the code, tends to infinity), the scheme in Theorem~\ref{thm:mainrs_small} is better for small $n$.

Our schemes are the first in this parameter regime to work for $r > 3$, and additionally they improve over previous work~\cite{multfail, multfail2} when $r = 2,3$.  More precisely, we obtain the following bounds:

\begin{itemize}
	\item Theorem~\ref{thm:mainrs_small} improves and generalizes the scheme of Dau et al. for Reed-Solomon codes in the centralized model \cite{multfail}.  More precisely, in Theorem~\ref{thm:mainrs_small}, for any $r \leq r_0$, for some $r_0 = O(\sqrt{\log(n)})$, we give schemes for high-rate (say, $1- \eps$), length $n$ Reed-Solomon codes which have repair bandwidth (measured in bits)
\[ \inparen{ (n - r) \cdot r - \frac{r(r-1)}{2}\inparen{ \frac{1}{\eps} - 1 } } \cdot \log_2(1/\eps). \]

	For comparison, the scheme of Dau et al. worked for $r=2,3$, and had bandwidth $(n-r) \cdot r \cdot \log_2(1/\eps)$. Thus, for $r = 2,3$ Theorem~\ref{thm:mainrs_small} improves the bandwidth by $\frac{r(r-1)}{2}(1/\eps - 1)\log_2(1/\eps)$ bits, and for larger $r$ we give the first non-trivial schemes for Reed-Solomon Codes with small subpacketization.

	When $r = 1$, this collapses to the scheme of~\cite{gw16}, which is optimal.
	
	\item Theorem~\ref{thm:mainrs_large} 
improves over Theorem~\ref{thm:mainrs_small} asymptotically, but it is not as good for small values of $n$.
For this theorem we again generalize the constructions of \cite{gw16,DM17}, but we do so in a different way and do not go the same route as \cite{multfail}.
We obtain 
a repair scheme for high rate ($1-\epsilon$), length $n$ Reed Solomon codes with repair bandwidth (in bits) at most
	\[\min_{r' \geq r} \inparen{ (n - r') \inparen{ \log_{2}(n) - \left\lfloor \log_{2}\inparen{ \frac{n\epsilon + r' -1 }{ 2r' - 1 } }\right\rfloor } }\]
	
	When $r = 1$, this too collapses to the scheme of~\cite{gw16}, which is optimal, and the bound is nontrivial for $r \leq r_0$ where $r_0 = O(n^{1 - \eps})$.  
	
\end{itemize}

We compare our two bounds, along with the trivial bound of bandwidth $kt$ and the optimal bound for $r=1$, in Figure~\ref{fig:quantitative}.
\begin{figure}[h!]
\begin{center}
\includegraphics[width=7cm]{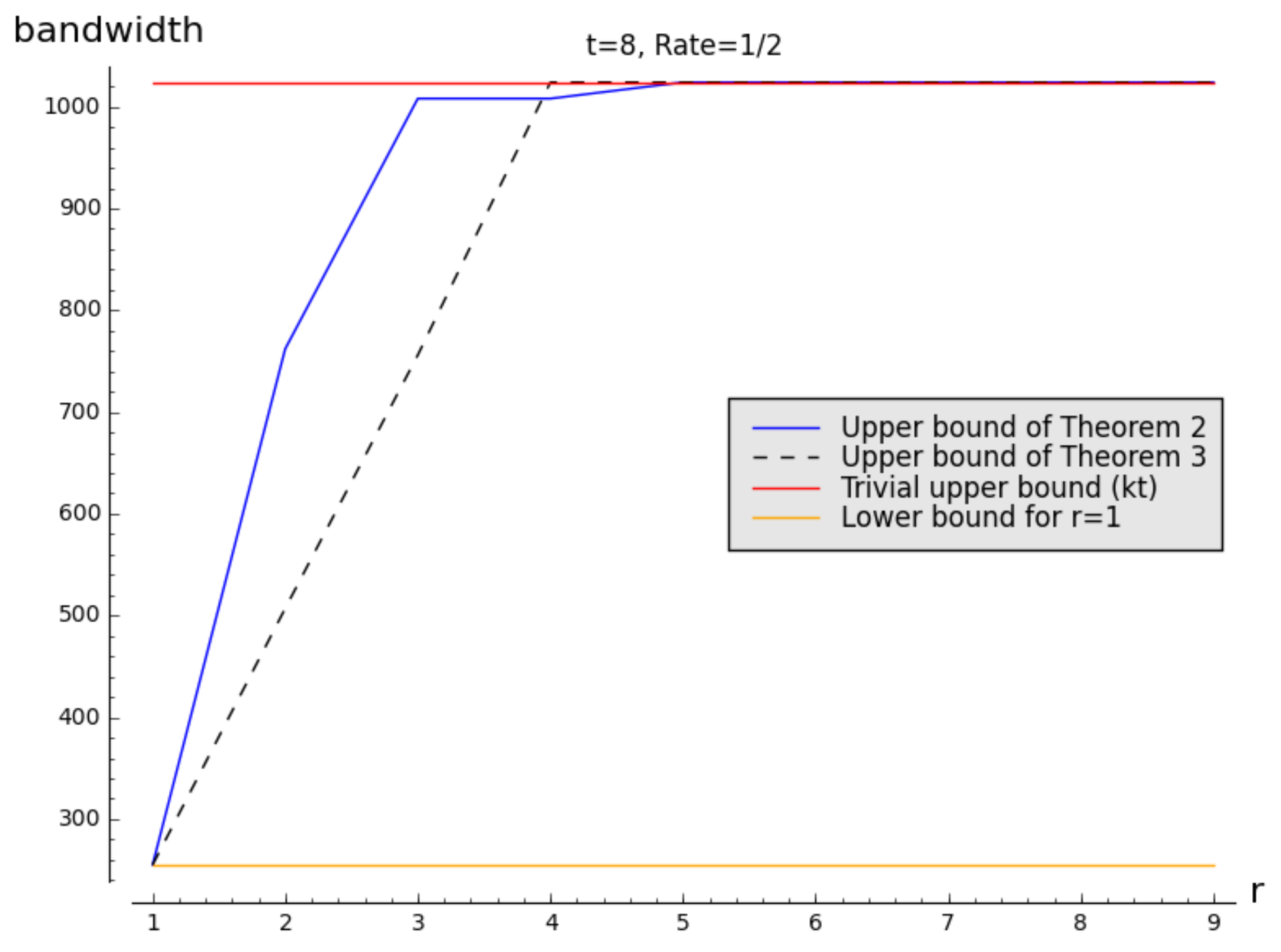}
\hspace{1cm} 
\includegraphics[width=7cm]{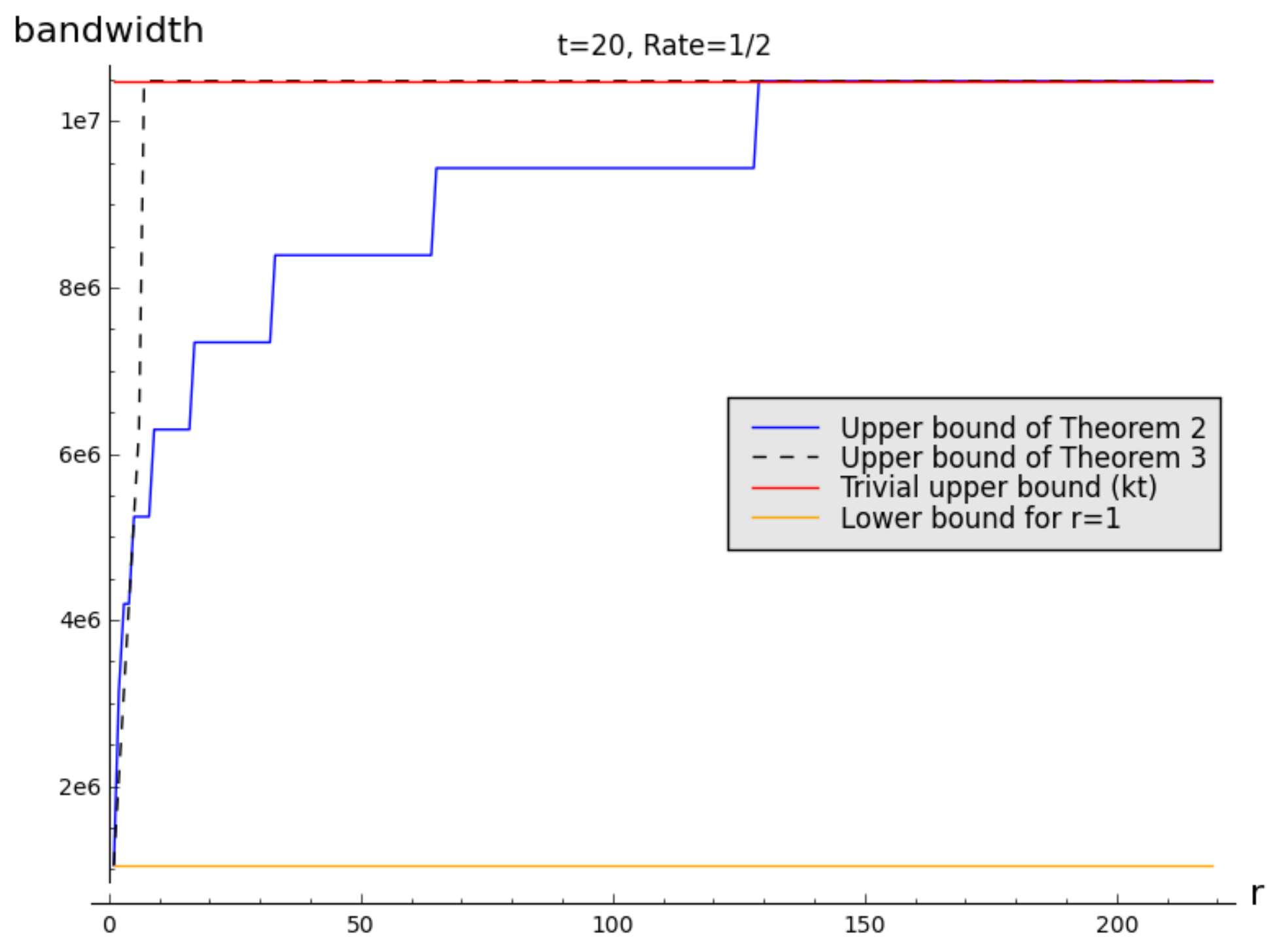} 
\end{center}
\caption{The bound of Theorems~\ref{thm:mainrs_large} and \ref{thm:mainrs_small}, for rate $1/2$ Reed-Solomon codes with $n = 2^t$, for $t=8$ and $t=20$.  When $t$ is small, Theorem~\ref{thm:mainrs_small} is better; when $t$ is large, Theorem~\ref{thm:mainrs_large} is better.  Other works~\cite{multfail,multfail2} on multiple failures in this parameter regime corrected up to two or three failures only. }
\label{fig:quantitative}
\end{figure}

\end{enumerate}
We emphasize that the codes in Theorems~\ref{thm:mainrs_large} and \ref{thm:mainrs_small} are simply Reed-Solomon codes that use all their evaluation points; that is, our results imply that this one classical code can be repaired from a growing number of failures with non-trivial bandwidth, and the repair behavior degrades gracefully as the number of failures increases to $n-k$.
However, we do not have a matching lower bound for larger $r$, and we suspect that further improvements are possible.

\paragraph{Organization.}
In Section \ref{sec:notation} we set up notation and give formal definitions for the problems we consider.
In Section~\ref{sec:framework}, we give Theorem~\ref{thm:main}, which provides a framework for constructing repair schemes for multiple failures for general scalar MDS codes.  In Section~\ref{sec:rs_r}, we give Theorems~\ref{thm:mainrs_large}~and~\ref{thm:mainrs_small}, which specialize Theorem~\ref{thm:main} to Reed-Solomon codes, and gives the results advertised above.

\section{Preliminaries}\label{sec:notation}
In this section, we set up notation, and formally introduce the definitions that we will work with throughout the paper.

\subsection{Notation}\label{ssec:notation}
We use the notation $[n]$ to mean the set of integers $\{1,\ldots,n\}$, and for vectors $v,w \in \F^n$, we use $\ip{v}{w} = \sum_{i \in [n]} v_i w_i$ to denote the standard inner product.

\paragraph{Matrix and vector notation.}
Unless otherwise noted, vectors $v$ are treated as column vectors; the $i$'th entry of a vector $v$ is denoted $v_i$.  For a vector $v \in \F^n$ and a set $I \subseteq [n]$, with $I = \{i_1,\ldots, i_r\}$ and $i_1 < i_2 < \cdots < i_r$, $v_I$ denotes the (column) vector $(v_{i_1}, v_{i_2}, \ldots, v_{i_r})$.
  For a vector $v \in \F^m$, we will use $\st{v}$ to denote the set
$ \st{v} = \inset{ v_i \suchthat i \in [m] }. $

For a matrix $M$, we use $M[:,i]$ to refer to the $i$'th column and $M[i,:]$ to refer to the $i$'th row of $M$.  For sets $I, J$, we will use $M[I,J]$ to refer to the submatrix of $M$ containing the rows indexed by $I$ and the columns indexed by $J$; and we will extend this to $M[I,:]$ and $M[:,J]$ to mean the submatrix formed by the rows in $I$ or columns in $J$, respectively.  Our notation is $1$-indexed.

\paragraph{Finite field notation.}
Throughout this paper, $\F$ denotes a finite field, and $B \subset \F$ denotes a subfield of $\F$.  
We use $\F^*$ and $B^*$ to denote the group of units in $\F$ and $B$ respectively, and $\F^*/B^*$ to denote the quotient group.
For a set of elements $S \subseteq \F$, we will use $\spn_B(S)$ to denote the linear span over $B$ of $S$:
\[ \spn_B(S) = \inset{ \sum_{x \in S} a_x \cdot x \suchthat a_x \in B }. \]
We will similarly use $\dim_B$ to refer to the dimension over $B$.
Finally, for a field $\F$ with a subfield $B$, so that $\F$ has degree $t$ over $B$, the field trace $\tr_{\F/B}: \F \to B$ is defined by
\[ \tr_{\F/B}(x) := \sum_{i=0}^{t-1} x^{|B|^i}. \]
The function $\tr_{\F/B}$ is a $B$-linear function from $\F$ to $B$.
We refer the reader to, for example, \cite{fraleigh} for a primer/refresher on finite fields.

\subsection{Definitions}
Let $\cC \subset \Sigma^n$ be a code of block length $n$ over an alphabet $\Sigma$.
As described in the introduction, we imagine the the $n$ symbols of a codeword $c = (c_1,c_2,\ldots,c_n) \in \cC$ are distributed between $n$ different nodes, so that node $i$ stores the symbol $c_i$. 

\paragraph{The exact repair problem.}
In the \em exact repair problem, \em one node, Node $i$, is unavailable, and the goal is to \em repair \em it (that is, recover $c_{i}$) using only information from the remaining nodes.  Of course, any MDS code can achieve this: by definition, all of $c$ is determined by any $k$ symbols, and so any $k$ surviving nodes determine all of $c$ and in particular the missing information $c_{i}$.  
But, as described in the introduction, we hope to do better than this, in terms of the amount of data downloaded.

Formally, suppose that $\Sigma \simeq B^t$ can is a vector space over some \em base field \em $B$.  
Thus, the contents of a node (a symbol $c_i \in \Sigma$) are $t$ sub-symbols from $B$.  When a node fails, a replacement node or repair center can contact a surviving node, which may do some computation and return some number---possibly fewer than $t$---sub-symbols from $B$.
The parameter $t$ is called the \em sub-packetization. \em
Formally, we define an exact repair scheme as follows.

\begin{definition}\label{def:exactrepair}
An \em exact repair scheme \em for a code $\cC \subset \Sigma^n$ is defined as follows.  For each $i \in [n]$, there is a collection of functions
\[ \inset{ g_{i,j} \suchthat j \in [n] \setminus \{i\} }, \]
so that 
\[ g_{i,j}: \Sigma \to B^{b_{i,j}} \]
for some non-negative integer $b_{i,j}$.
and so that for all $c \in \cC$, $c_i$ is determined from $\inset{ g_{i,j}(c_j): j \in [n] \setminus \{i\} }$.
The \em bandwidth \em of this scheme (measured relative to $B$) is the total number of elements of $B$ required to repair any node:
\[ \mathrm{bandwidth} = \max_{i \in [n]} \sum_{j \in [n] \setminus \{i\} } b_{i,j}. \]
\end{definition}

\begin{remark}[Variants]\label{rem:variants}
The definition above is not the only definition of regenerating codes, and is missing several parameters often considered.  For example, we may also limit \em number \em of nodes contacted, requiring the repair scheme to only contact $d$ out of the surviving nodes.  We may also allow for the nodes to store more elements of $B$ than the original data blocks to (in the lingo of regenerating codes, to move away from the MSR setting and toward the MBR setting).  However, the goal of the current work is to study multiple failures in scalar MDS codes.
\end{remark}

\paragraph{Multiple failures.}
In this work, we will focus on the centralized model of multiple repair~\cite{CJM13}. 
In this model, a repair center is in charge of the repair for all the nodes.  We count as bandwidth the information downloaded by this repair center, but not between the center and any of the replacement nodes.  Formally, we have the following definition.

\begin{definition}\label{def:exactrepairmult}
An \em exact centralized repair scheme \em for $r$ failures for a code $\cC \subset \Sigma^n$ is defined as follows.  For each set $I \subseteq [n]$ of size at most $r$, there is a collection of functions
\[ \inset{ g_{I,j} \suchthat j \in [n] \setminus I }, \]
so that 
\[ g_{I,j}: \Sigma \to B^{b_{I,j}} \]
for some non-negative integer $b_{I,j}$,
and so that for all $c \in \cC$, and all $i \in I$, $c_i$ is determined from $\inset{ g_{I,j}(c_j): j \in [n] \setminus I }$.
The \em bandwidth \em of this scheme (measured relative to $B$) is the total number of elements of $B$ required to repair the nodes in any set $I$:
\[ \mathrm{bandwidth} = \max_{I \subset [n], |I| \leq r} \sum_{j \in [n] \setminus \{i\} } b_{I,j}. \]
\end{definition}

Definition~\ref{def:exactrepairmult} is perhaps the simplest possible definition of the exact repair problem for multiple failures.  As per Remark~\ref{rem:variants}, we could spice up the definition of the exact repair problem in many ways; and beyond that following the work of~\cite{multfail} for Reed-Solomon codes in other models, we could include in our measure of bandwidth some way to capture the cost of communication between the multiple replacement nodes.  However, addressing even this simplest case is interesting and much is unknown, so we will focus on this case for the current work, and we hope that the insights of this work may be extended to more complicated models.

\paragraph{Linear repair schemes and scalar MDS codes.}
As mentioned in the introduction, most of the work on regenerating codes explicitly views the alphabet $\Sigma$ as a vector space over some field $B$.  However, for many codes commonly used in distributed storage---notably Reed-Solomon codes---it is more common to view the alphabet $\Sigma$ as a finite field $\F$.  Such codes are termed ``scalar" MDS codes~\cite{scalarmds}.  However, if $B \subseteq \F$ is a subfield so that the degree of $\F$ over $B$ is $t$, then $\F$ is in fact a vector space of dimension $t$ over $B$, and so the set-up above makes sense.
We focus on this setting for the rest of the paper: that is, $\cC \subset \F^n$ is a linear subspace which has the property that any $k$ symbols of a codeword $c \in \cC$ determine $c$.

In this setting, while more restrictive\footnote{The difference is that an array code with the MDS property need not be itself a linear code over $\Sigma$ (and indeed this may not even make sense if $\Sigma$ is not a field), while a scalar MDS code is by definition linear over $\Sigma$.} than that of Definition~\ref{def:exactrepair}, there is additional algebraic structure which, it turns out, very nicely characterizes exact repair schemes for a scalar MDS code $\cC \subset \F^n$ (for a single failure) in terms of the \em dual code \em $\cC^\perp := \inset{ v \in \F^n \suchthat \ip{c}{v} = 0 \forall c \in \cC}$.  More formally, we define a \em repair matrix \em for a symbol $i \in [n]$ as follows.
\begin{definition}\label{def:repmat}
Let $\cC \subseteq \F^n$ be an  MDS code over $\F$, and suppose that $B$ is a subfield of $\F$, so that $\F$ has degree $t$ over $B$.  Let $i \in [n]$.  A \em repair matrix \em with bandwidth $b$ for an index $i$ is a matrix $M \in \F^{n \times t}$ with the following properties:
\begin{enumerate}
	\item The columns of $M$ are codewords in the dual code $\cC^\perp$.
	\item The elements of the $i$'th row $M[i,:]$ of $M$ have full rank over $B$.
	\item We have
\[ \sum_{j \in [n]\setminus\{i\}} \dim_B \inset{ \spn_B\inset{ \st{M[j,:]} } } = b. \]
\end{enumerate}
\end{definition}
One of the main results of \cite{gw16} was that repair matrices precisely characterize linear repair schemes.  We say that a repair scheme as in Definition~\ref{def:exactrepair} is \em linear \em if the functions $g_{i,j}$, along with the function that determines $c_i$, are all $B$-linear.   The work of \cite{gw16} showed that a (scalar) MDS code $\cC$ admits a linear repair scheme with bandwidth $b$ if and only if, for all $i \in [n]$, there is a repair matrix with bandwidth at most $b$ for $i$.

\section{Framework}\label{sec:framework}
In this section, we extend the framework of \cite{gw16} to the case of multiple repairs.  
We first define an analog of repair matrices for multiple repair. 
\begin{definition}\label{def:megarepmat}
Let $\cC \subseteq \F^n$ be a MDS code over $\F$, and suppose that $B$ is a subfield of $\F$, so that $\F$ has degree $t$ over $B$.  Let $I \subset [n]$ have size $r$.  A \em multiple-repair matrix \em with bandwidth $b$ for $I$ is a  matrix $M \in \F^{n \times rt}$ with the following properties:
\begin{enumerate}
        \item\label{prop:dual} The columns of $M$ are codewords in the dual code $\cC^\perp$.
	\item\label{prop:fullrank} The submatrix $M[I,:]$ has full rank over $B$, in the sense that for all nonzero $x \in B^{rt}$, $M[I,:] \cdot x \neq 0$.
        \item\label{prop:lowdim} We have
\[ \sum_{j \in [n]\setminus I} \dim_B \inset{ \spn_B\inset{ \st{M[j,:]} } } = b. \]
\end{enumerate}
\end{definition}

Our main theorem is that an MDS code $\cC$ admits a (linear) repair scheme for a set $I$ of failed nodes with bandwidth $b$ if there exists a multiple-repair matrix with bandwidth $b$ for $I$.
\begin{theorem}\label{thm:main}
Let $\cC \subset \F^n$ be an MDS code, and let $B \subset \F$ be a subfield so that $\F$ has degree $t$ over $B$.  
Suppose that for all $I \subseteq [n]$ of size $r$, there is a multiple-repair matrix $M \in \F^{n \times rt}$ with bandwidth at most $b$ for $I$.
Then $\cC$ admits an exact centralized repair scheme for $r$ failures with bandwidth $b$.
\end{theorem}
\begin{proof}
Let $I \subset [n]$ be any set of $r$ failures, and let $M \in \F^{n \times rt}$ be a multiple-repair matrix with bandwidth $b$ for $I$.  
For each $j \in [n] \setminus I$, 
we will show how to use $M$ to construct the functions $g_{I,j}: \F \to B^{b_{I,j}}$. 

We will choose $b_{I,j}$ (the number of sub-symbols returned by $g_{I,j}$) to be $b_{I,j} = \dim_B\inset{ \spn_B \inset{ \st{M[i,:]} } }$. Then by Definition~\ref{def:megarepmat}, $\sum_{j \in [n] \setminus I} b_{I,j} \leq b$.  Let $\lambda_1,\ldots, \lambda_{b_{I,j}} \in \F$ be a basis for the elements of $M[j,:]$ over $B$.
(We note that the $\lambda_i$ depend on the choice of $j$, but we suppress this for notational clarity).
For $x \in \F$, we choose
\[ g_{I,j}(x) = ( \tr_{\F/B}(\lambda_1 \cdot x), \tr_{\F/B}(\lambda_2 \cdot x), \cdots, \tr_{\F/B}(\lambda_{b_{I,j}} \cdot x ) ). \]
We first observe that, by Property~\ref{prop:lowdim} in Definition~\ref{def:megarepmat}, the total bandwidth of this scheme is $b$ symbols of $B$.
We next need to show that this repair scheme works; that is,
we need to show that for all $c \in \cC$, the values $\inset{g_{I,j}(c_j) \suchthat j \in [n] \setminus I}$ determine $\inset{c_i \suchthat i \in I }$.

By Property \ref{prop:dual} in Definition~\ref{def:megarepmat}, for all $\ell \in [rt]$, we have $M[:,\ell] \in \cC^\perp$.  
This means that for all $c \in \cC$, and for all $\ell \in [rt]$, 
\begin{align*}
0 &= \sum_{i \in [n]} c_i \cdot M[i,\ell] \\
\sum_{i \in I} c_i \cdot M[i,\ell] &= -\sum_{j \in [n] \setminus I } c_j \cdot M[j,\ell] \\ 
\tr_{\F/B} \inparen{ \sum_{i \in I} c_i \cdot M[i,\ell] } &= \tr_{\F/B} \inparen{ -\sum_{j \in [n] \setminus I} c_j \cdot M[j,\ell] } \\
\sum_{i \in I} \tr_{\F/B} (c_i \cdot M[i,\ell] ) &= -\sum_{j \in [n] \setminus I} \tr_{\F/B} (c_j \cdot M[j,\ell] ).
\end{align*}
We claim that the right-hand side above can be constructed from the values $\inset{g_{I,j}(c_j) \suchthat j \in [n] \setminus I}$.   Indeed, write $M[j,\ell] = \sum_{i = 1}^{b_{I,j}} a_{i,\ell,j} \lambda_i$ for some coefficients $a_{i,\ell,j} \in B$.  Then,
\begin{align*}
-\sum_{j \in [n] \setminus I} \tr_{\F/B} (c_j \cdot M[j,\ell]) 
&= -\sum_{j \in [n] \setminus I} \tr_{\F/B} \inparen{ c_j \cdot \sum_{i=1}^{b_{I,j}} a_{i,\ell,j} \lambda_i } 
= -\sum_{j \in [n] \setminus I} \sum_{i=1}^{b_{I,j}} a_{i,\ell,j} \tr_{\F/B} \inparen{ c_j \cdot \lambda_i },
\end{align*}
and the values $\tr_{\F/B} \inparen{ c_j \cdot \lambda_i }$ are precisely what is returned by $g_{I,j}(c_j)$.  
Thus, given the returned information, the repair center can reconstruct the quantities
\begin{equation}\label{eq:whatweget}
	\sum_{i \in I} \tr_{\F/B} (c_i \cdot M[i,\ell])  \qquad \forall \ell \in [rt].
\end{equation}

Finally, we invoke Property~\ref{prop:fullrank} in Definition~\ref{def:megarepmat} to show that \eqref{eq:whatweget} in fact contain enough information to recover $\inset{ c_i \suchthat i \in I}$.  To see this, consider the map
$ \vphi: \F^r \to B^{rt} $
given by
\[ \vphi(x) = \tr_{\F/B}\inparen{ x^T \cdot M[I,:] }, \]
where the multiplication is done over $\F$ and the trace is applied entry-wise.
That is,
\[ \vphi(x) = ( \tr_{\F/B}(\ip{ x }{ M[I,1] }), \tr_{\F/B}( \ip{x}{M[I,2]} ), \ldots, \tr_{\F/B}( \ip{x}{M[I,rt]} ) ). \]
We will show that $\vphi$ is invertible.  
To see this, consider the map
$ \psi: B^{rt} \to \F^r $
given by
\[ \psi(y) = M[I,:] \cdot y. \]
This map is clearly $B$-linear and 
Property~\ref{prop:fullrank} says that $\psi$ is injective.
By counting dimensions (over $B$), $\psi$ is surjective as well.  
To conclude, we will observe that $\psi$ is the adjoint of $\vphi$, in the sense that for all $y \in B^{rt}$ and for all $x \in \F^r$, we have
\[ \ip{ \vphi(x) }{y} = \tr_{\F/B}\inparen{ \ip{x }{\psi(y) }}, \]
and hence since $\psi$ is invertible then $\vphi$ is invertible.
Formally, we compute
\begin{align*}
\ip{\vphi(x) }{y} &= \sum_{j \in [rt]} y_j \cdot \tr_{\F/B}( \ip{x}{M[I,j]} ) \\
&= \tr_{\F/B}\inparen{ \sum_{j \in [rt]} y_j \cdot \ip{x}{M[I,j]}  } \\
&= \tr_{\F/B}\inparen{ \ip{x}{ M[I,:]\cdot y} } \\
&= \tr_{\F/B} \inparen{ \ip{x}{ \psi(y) } }.
\end{align*}
Now, we would like to show that $\vphi$ is injective.  Let $x \in \F^r$ be nonzero.  Then there is some $z \in \F^r$ so that $\tr_{\F/B}(\ip{x}{z}) \neq 0$.  Because $\psi$ is surjective, there is some $y \in B^{rt}$ so that $\psi(y) = z$.  But then
\[ \ip{\vphi(x)}{y} = \tr_{\F/B}(\ip{x}{\psi(y)}) = \tr_{\F/B}(\ip{x}{z}) \neq 0, \]
and hence $\vphi(x) \neq 0$ as well.  This shows that $\vphi$ is injective; again by dimension counting, we see that $\vphi$ is also surjective and hence invertible.

Thus, given $\vphi(x)$, we may recover $x$ via linear algebra.
To complete the argument, we observe that the quantities \eqref{eq:whatweget} in fact give us $\vphi(c_I)$, where we recall that $c_I$ denotes the restriction of $c$ to $I$.  Thus, given \eqref{eq:whatweget}, we may invert $\vphi$ and recover $\inset{c_i \suchthat c \in I}$, as desired.
\end{proof}

\section{Centralized repair schemes for RS codes with multiple failures}\label{sec:rs_r}
In this section we specialize Theorem~\ref{thm:main} to Reed-Solomon codes.
The \em Reed-Solomon Code \em $\cC$ of dimension $k$ over $\F$ with evaluation points $\alpha_1,\ldots,\alpha_n$ is the set
\[ \cC = \inset{ (f(\alpha_1), f(\alpha_2), \ldots, f(\alpha_n) ) \suchthat f \in \F[X], \deg(f) < k }. \]
The dual of any Reed-Solomon Code is a \em generalized Reed-Solomon Code: \em
\begin{equation}\label{eq:grs}
 \cC^\perp = \inset{ (  \lambda_1 \cdot g(\alpha_1),  \lambda_2 \cdot g(\alpha_2), \ldots,  \lambda_n \cdot g(\alpha_n) ) \suchthat g \in \F[X], \deg(g) < n - k },
\end{equation}
where $\lambda_1,\ldots, \lambda_n \in F$ are constants which depend on the choice of $\alpha_1,\ldots, \alpha_n$.  (We refer the reader to\cite{macwilliams-sloane} 
for more details).

Below, we give two constructions instantiating Theorem~\ref{thm:main} for Reed-Solomon Codes.  Our first scheme, discussed in Section~\ref{sec:rs_r_large}, is much better as $n \to \infty$.  However, for small values of $n$, our second scheme, discussed in Section~\ref{sec:rs_r_small}, is quantitatively better, and so we include it as it may be of more practical interest.  

Both of our schemes generalize the construction of \cite{gw16,DM17} for a single failure.
We briefly review this construction below, as we will need it for our constructions.
The scheme of \cite{gw16} was as follows:
\begin{proposition}[\cite{gw16}]\label{prop:repairscheme}
	Let $n = |\F|$, and let
	Let $\cC \subseteq \F^n$ be the Reed-Solomon code of dimension $k = n - n/|B|$, which uses all evaluation points $\F = \inset{ \alpha_1,\ldots,\alpha_n }$.
	Let $\delta \in \F$, and let $\zeta_1,\ldots,\zeta_t$ be a basis for $\F$ over $B$.
	Then the matrix $M = M(\delta, \zeta_1,\ldots,\zeta_t)  \in \F^{n \times t}$ with
	\begin{equation}\label{eq:repairscheme}
	M[j,w] = \frac{ \delta \cdot \tr( \zeta_w (\alpha_j - \alpha_{i}) ) }{ \alpha_j - \alpha_{i} } 
	\end{equation}
	is a repair matrix for index $i$ with bandwith $n-1$ symbols of $B$. 
\end{proposition}
To see that this is indeed a valid repair matrix for $i$, observe that the polynomial
\[ h_w(X) = \frac{ \delta \cdot \tr_{\F/B}( \zeta_w (X - \alpha_{i}) ) }{ X - \alpha_{i} }
= \delta \inparen{ \zeta_w + \zeta_w^{|B|}(X - \alpha_{i})^{|B| - 1} + \cdots + \zeta_w^{|B|^{t-1}} (X - \alpha_{i})^{|B|^{t-1} -1} } \]
is indeed a polynomial of degree less than $n - k = n/|B| = |B|^{t-1}$, and so the column $M[:,j]$ is an element of $\cC^\perp$.  Moreover, we have $h_w(\alpha_{i}) = \delta \zeta_w$, and so $M[i,:] = ( \delta \zeta_1, \delta \zeta_2,\ldots, \delta \zeta_t )$ is full rank.  Finally, for all $j \neq i$, 
\[ h_w(\alpha_j) \in \frac{ \delta }{\alpha_j - \alpha_{i} } \cdot B, \]
and hence $\spn_B( \st{M[j,:]} ) = \spn_B\inset{ h_w(\alpha_j) \suchthat w \in [t] }$
has dimension $1$ over $B$, and so the bandwidth of the repair matrix is $n-1$.

In \cite{DM17}, it was observed that the trace function above can be replaced with an arbitrary linearized polynomial.
More precisely,
let $W \subseteq F$ be a subspace of dimension $s$ over $B$.  The \em subspace polynomial \em defined by $W$ is
\[ L_W(X) = \prod_{\alpha \in W}(X - \alpha). \]
It is well-known that $L_W$ is a $B$-linear map from $F$ to $F$, of the form
\begin{equation}\label{eq:linearizedform}
 L_W(X) = \sum_{j=0}^s c_j X^{|B|^j} 
\end{equation}
for coefficients $c_0,\ldots,c_{s} \in F$.  In particular, there is no constant term; $X$ divides $L_W(X)$.
Moreover, the coefficient $c_0$ is nonzero, as
\[ c_0 = \sum_{\alpha \in W} \prod_{ \beta \in W\setminus\{\alpha\} } \beta = \prod_{\beta \in W \setminus \{0\}} \beta \neq 0. \]
  Since the kernel of $L_W$ is $W$, the image of $L_W$ is a subspace of dimension $t-s$ over $B$.  

With this background in place, we proceed to our constructions.

\subsection{Main construction}\label{sec:rs_r_large}

Our main construction for Reed-Solomon codes generalizes the construction of~\cite{gw16,DM17}.  In particular, we will choose $rt$ different low-degree polynomials of a form similar to \eqref{eq:repairscheme}; we follow \cite{DM17} and replace the trace function with a linearized polynomial.  The key is to choose an appropriate modification so that the requirements of Theorem~\ref{thm:main} hold. 
We will see how to do this below, but first we state our main result.

\begin{theorem}\label{thm:mainrs_large}
Let $B \subseteq F$ be a subfield, and let $\cC \subset F^n$ be a Reed-Solomon Code of dimension $k$.  
Choose $1 \leq r \leq n- k$.
Then $\cC$ admits an exact centralized repair scheme for $r$ failures with bandwidth at most
\begin{equation}\label{eq:ourresult}
 b \leq \min_{r' \geq r} \inparen{ (n - r') \inparen{ t - \left\lfloor \log_{|B|}\inparen{ \frac{n - k + r' -1 }{ 2r' - 1 } }\right\rfloor } }. 
\end{equation}
\end{theorem}

\begin{remark}[Optimality?]
Theorem ~\ref{thm:mainrs_large} works for any value of $r$ between $1$ and $n-k$, and moreover it is optimal for scalar MDS codes at both $r=1$ and $r=n-k$.
Moreover, it produces non-trivial bounds for $r \ll n^R$, where $R = k/n$ is the rate of the code.  However, when $r \approx n^R$, the expression above becomes trivial (equal to $kt$), and we do not know whether or not this is a fundamental limit or an artifact of our approach.
The behavior of \eqref{eq:ourresult} is shown in Figure~\ref{fig:quantitative}.
\end{remark}

\begin{proof}[Proof of Theorem~\ref{thm:mainrs_large}]
Let $W \subseteq F$ be a subspace of dimension $s$ over $B$, and let $L_W$ be the subspace polynomial 
\[ L_W(X) = \prod_{\alpha \in W}(X - \alpha). \]
As above, let $I \subset [n]$ be the set of failed nodes, so $|I| = r$.  
Let $Z = \inset{ \zeta_1,\ldots, \zeta_t}$ be a basis for $F$ over $B$.
Define
\[ F_I(X) = \prod_{\alpha \in I}(X - \alpha),\]
and for $p \in [r]$, define
\[ P_{\zeta,p}(X) = \frac{ L_W\inparen{ \zeta \cdot F_I(X) \cdot X^{p-1} } }{F_I(X)}.\]
This is a generalization of the construction in \cite{gw16,DM17} to multiple failures: if $r = 1$ and $\alpha_1$ is the only failed node, then the above is
\[ P_{\zeta,1}(X) = \frac{ L_W\inparen{ \zeta \cdot (X - \alpha_1) } }{(X - \alpha_1)}, \]
exactly as in \cite{DM17}.  

Now, we construct our repair matrix $M \in F^{n \times rt}$ as follows.  We index the rows of $M$ by $j \in [n]$, and the columns by pairs $(\zeta, i)$ for $\zeta \in Z$ and $p \in [r]$.
Then we set
\[ M[j, (\zeta,p)] := \lambda_j P_{\zeta,p}(\alpha_j), \]
where $\lambda_j$ is as in \eqref{eq:grs}.  

We must show that $M$ satisfies the conditions of Theorem~\ref{thm:main}.  
In particular, we will show that $M$ has the form shown in Figure~\ref{fig:construction}, which as we will see below implies that Theorem~\ref{thm:main} applies.
First, we compute the bandwidth of $M$.
\begin{claim}\label{claim:bw}
The bandwidth of $M$ is at most $(n-r)(t-s)$.
\end{claim}
\begin{proof}
For $j \not\in I$, the set of symbols $\st{M[j,:]}$ that appears in row $j$ of $M$ is precisely
\begin{align*}
 \st{M[j,:]} &= \inset{ \lambda_j \cdot p_{\zeta, p}(\alpha_j) \,:\, \zeta \in Z, p \in [r] }\\
&= \inset{ \lambda_j \cdot \frac{ L_W( \zeta \cdot F_I(\alpha_j) \cdot \alpha_j^{p-1}) }{ F_I(\alpha_j) } \,:\, \zeta \in Z, p \in [r] } \\
&\subseteq \frac{\lambda_j}{F_I(\alpha_j)} \Img( L_W ),
\end{align*}
which is a subspace of dimension $t - s$.  Since there are $n-r$ such $j$'s by definition the bandwidth of $M$ is bounded by
\[ (n-r)(t-s), \]
as desired.
\end{proof}
We next show that $M[I,:]$ has full rank, in the sense that $M[I,:] x \neq 0$ for all nonzero $x \in B^{rt}$.
\begin{claim}
$M[I,:]$ has full rank.
\end{claim}
\begin{proof}
Using \eqref{eq:linearizedform}, we write
\begin{align*}
 P_{\zeta, p}(X) &= \frac{ P_W( \zeta \cdot F_I(X) \cdot X^{p-1}) }{F_I(X) } \\
&= \frac{ \sum_{m=0}^s c_m \inparen{ \zeta \cdot F_I(X) \cdot X^{p-1} }^{|B|^m} }{ F_I(X) } \\
&= \sum_{m=0}^s c_m \cdot \zeta^{|B|^m} \cdot F_I(X)^{|B|^m - 1} \cdot X^{(p-1)\cdot |B|^m},
\end{align*}
and hence for $i \in I$, all terms vanish except the $m=0$ term, and we have
\[ P_{\zeta, p}(\alpha_i) = c_0 \cdot \zeta \cdot \alpha_i^{p-1}. \]

To show that $M[I,:]$ has full rank, by definition we must show that for all nonzero $x \in B^{rt}$, $M[I,:]x \neq 0$.  So let $x \in B^{rt} \setminus \{0\}$, and write $x = (x^{(1)}, x^{(2)}, \ldots, x^{(r)})$, where each $x^{(i)} \in B^t$.  By the above characterization, we have
\[ \ip{M[i,:]}{x} = \lambda_i \cdot \sum_{p=1}^r c_0 \alpha_i^{p-1} \ip{\vec{\zeta}}{x^{(p)}}, \]
where $\vec{\zeta} = (\zeta_1,\ldots,\zeta_t) \in F^t$.
Since the $\zeta_i$ form a basis of $F$ over $B$ and since $x \neq 0$, at least one of the coefficients $\ip{\vec{\zeta}}{x^{(p)}}$ is nonzero.  Thus, $\ip{M[i,:]}{x} = \lambda_i \cdot g(\alpha_i)$, where $g(X) = \sum_{p=1}^r c_0 \ip{\vec{\zeta}}{x^{(p)}} X^{p-1}$ is some nonzero polynomial of degree at most $r-1$.  This means that $M[i,:] \cdot x \neq 0$, or else this polynomial $g(X)$ would have at least $r$ roots, one for each $i \in I$.
\begin{figure}[h!]
\footnotesize
\begin{tikzpicture}[xscale=1.2,yscale=0.7]

\draw[thick] (0,0) rectangle (12,-8);
\draw (0,-1) to (12, -1);
\draw (0,-2) to (12, -2);
\draw[thick] (0, -3) to (12, -3);
\draw (0, -4) to (12, -4);
\draw (0, -5) to (12, -5);
\node at (6, -6) {$\cdots$};
\draw (0, -7) to (12, -7);
\draw (4,0) to (4,-3);
\draw (8,0) to (8,-3);
\node at (2,-.5) {$ \lambda_{i_1} \cdot c_0 \cdot \vec{\zeta}$};
\node at (6,-.5) {$\lambda_{i_1}\cdot c_0 \cdot \alpha_{i_1} \cdot \vec{\zeta}$};
\node at (10, -.5) {$\lambda_{i_1}\cdot c_0 \cdot \alpha_{i_1}^2 \cdot \vec{\zeta}$};
\node at (2,-1.5) {$\lambda_{i_2}\cdot c_0 \cdot \vec{\zeta}$};
\node at (6,-1.5) {$\lambda_{i_2}\cdot c_0 \cdot \alpha_{i_2} \cdot \vec{\zeta}$};
\node at (10, -1.5){$\lambda_{i_2}\cdot c_0 \cdot \alpha_{i_2}^2 \cdot \vec{\zeta}$};
\node at (2,-2.5) {$\lambda_{i_3}\cdot c_0 \cdot \vec{\zeta}$};
\node at (6,-2.5) {$\lambda_{i_3}\cdot c_0 \cdot \alpha_{i_3} \cdot \vec{\zeta}$};
\node at (10, -2.5) {$\lambda_{i_3}\cdot c_0 \cdot \alpha_{i_3}^2  \cdot \vec{\zeta}$};

\node at (6, -3.5) {$ \in \lambda_{j_1} \cdot (F_I(\alpha_{j_1}))^{-1} \cdot \Img(L_W) $ };
\node at (6, -4.5) {$ \in \lambda_{j_2} \cdot (F_I(\alpha_{j_2}))^{-1} \cdot \Img(L_W) $ };
\node at (6, -7.5) {$ \in \lambda_{j_3} \cdot (F_I(\alpha_{j_{n-3}}))^{-1} \cdot \Img(L_W) $ };

\draw [decorate,decoration={brace,amplitude=10pt,mirror},xshift=-4pt,yshift=0pt]
(0,0) -- (0,-3)node [black,midway,xshift=-11pt,anchor=east] {$M[I,:]$};
\draw [decorate,decoration={brace,amplitude=10pt},yshift=4pt,xshift=0pt]
(0,0) -- (4,0)node [black,midway,yshift=11pt,anchor=south] {$t$};
\draw [decorate,decoration={brace,amplitude=10pt,mirror},xshift=-4pt,yshift=0pt]
(0,-3) -- (0,-8)node [black,midway,xshift=-11pt,anchor=east] {$M[I^c,:]$};

\end{tikzpicture}
\normalsize
\caption{The matrix $M$ constructed in the proof of Theorem~\ref{thm:mainrs_large} for $r = 3$, where we write $I = \{i_1,i_2,i_3\}$, and $[n] \setminus I = \{j_1,\ldots, j_{n-3}\}$.
Here, $c_0$ is the coefficient from \eqref{eq:linearizedform}, and $\zeta = (\zeta_1,\ldots,\zeta_r)$ is a vector consisting of the elements of the basis $Z$.
}\label{fig:construction}
\end{figure}
\end{proof}

Finally, we must choose $s$.  We need for the columns of $M$ to be elements of $\cC^\perp$, which by \eqref{eq:grs} is the same as requiring the polynomials $P_{\zeta, p}(X)$ to have degree strictly less than $n - k$.
That is, we require
\[ |B|^s \cdot (2r - 1) - r \leq n - k - 1, \]
which is satisfied by the choice of 
\[ s = \left\lfloor \log_{|B|}\inparen{ \frac{ n - k + r - 1 }{ 2r - 1 } } \right\rfloor. \]
Plugging this choice of $s$ into the bandwidth bound of Claim~\ref{claim:bw} coupled with Observation~\ref{obs:CanUseMin} below finishes the proof of Theorem~\ref{thm:mainrs_large}.

\begin{observation}\label{obs:CanUseMin}
	Given $r' > r$ if $\cC$ admits an exact centralised repair scheme for $r'$ failures with bandwidth $b$ then it also admits an exact centralized repair scheme for $r$ failures with bandwidth $b$.
\end{observation}
\begin{proof}
	This is true because if we have $r < r'$ failures then we can simply ignore a further $r'-r$ nodes and use the repair scheme for $r'$ failures.
\end{proof}

\end{proof}

\subsection{Alternate construction for small block sizes}\label{sec:rs_r_small}
In this section we give another generalization of the one-failure scheme from Proposition~\ref{prop:repairscheme}. 
This scheme is worse asymptotically, but has better performance for small $n$, so it may be of practical value.
The basic idea directly generalizes (and improves upon) that of \cite{multfail2}; the multiple-repair matrix $M$ is formed by concatenating $r$ separate repair matrices $M_1,\ldots,M_r$ from Proposition~\ref{prop:repairscheme}.  In fact, Theorem~\ref{thm:main} immediately implies that this is a nontrivial repair scheme, but we can do better by choosing multipliers $\delta_1,\ldots,\delta_r \in \F$, and using the repair matrix formed by concatenating $\delta_1 M_1, \ldots, \delta_r M_r$.  We will show how to choose the multipliers $\delta_1,\ldots,\delta_r$ so that (a) the rank of $M[I,:]$ is not affected, but (b) the rank of the other rows $M[j,:]$ for $j \not\in I$ is reduced.

We will prove the following theorem.
\begin{theorem}\label{thm:mainrs_small}
	Let $n = |\F|$ and let $\cC \subseteq \F^n$ be as in Lemma~\ref{thm:sufficient}.  Let $B$ be a subfield of $\F$ so that $\F$ has degree $t$ over $B$.
	Choose $r \geq 2$ so that
	\[t > {r \choose 2} + \log_{|B|}(r\cdot(r+{r \choose 2}(|B|-1))\cdot (|B|-1) + 1).\]
	Then for all $I \subset [n]$ of size $r$, there is a matrix $M \in \F^{n \times rt}$ so that $M$ is a multiple repair matrix for $I$, with bandwidth 
	\[ b \leq (n-r)\cdot r - (|B| - 1){r \choose 2}. \]
\end{theorem}
\begin{remark}[Bandwidth guarantee]
	Observe that the naive scheme (which contacts any $k$ remaining nodes) has bandwidth $nk$, while the scheme which repeats the one-failure scheme $r$ times has bandwidth $(n-1)\cdot r$.   Thus, the guarantee that $b \leq (n-r)\cdot r - {r \choose 2}(|B| - 1)$ improves on both of these.  Moreover, when $r=1$, this collapses to the result of \cite{gw16} that $b \leq n-1$.   For $r = 2,3$, this improves over the result $b \leq (n-r) \cdot r$ of \cite{multfail}.  A comparison of Theorem~\ref{thm:mainrs_small} with Theorem~\ref{thm:mainrs_large} is shown in Figure~\ref{fig:quantitative}.
\end{remark}
\begin{remark}[Large $r$]
	Notice that Theorem~\ref{thm:mainrs_small} allows $r$ to grow slightly with $n$.  However, since we have $t =\log_{|B|}(n)$ since $n = |\F|$, the requirement on $t$ implies that, for the result to hold, we need
	\[ \log_{|B|}\inparen{ \frac{n}{r\cdot(r+{r \choose 2}(|B|-1))\cdot (|B|-1) + 1} } > {r \choose 2} \]
	or $r \lesssim \sqrt{\log(n)}$.  
\end{remark}

The rest of this section is devoted to the proof of Theorem~\ref{thm:mainrs_small}.  We begin with a lemma which 
shows that, if the multipliers $\delta_1,\ldots,\delta_r$ are picked appropriately, then the matrix formed by contenating $r$ copies of the single-repair matrices of Proposition~\ref{prop:repairscheme} form a good multiple-repair matrix. 
\begin{lemma}\label{thm:sufficient}
	Let $n = |\F|$, and let $B \subseteq \F$ be a subfield so that $F$ has degree $t$ over $B$. 
	Let $\cC \subseteq \F^n$ be the Reed-Solomon code of dimension $k = n - n/|B|$ with evaluation points $\F = \{\alpha_1,\ldots,\alpha_n\}$.
	Suppose\footnote{Since we will never use anything about the ordering of the evaluation points, this assumption is without loss of generality.} that $I = \inset{ \alpha_1,\ldots,\alpha_r }$.  
	Let $\zeta_1,\ldots,\zeta_t$ be any basis for $\F$ over $B$.  
	Choose $\delta_1,\ldots, \delta_r$ so that for all $j = 1,\ldots, r$, for all $\ell > j$ and for all $s > j$, we have
	\begin{equation}\label{eq:needthis}
	\tr_{\F/B}\inparen{ \frac{ \delta_\ell }{\delta_j} \cdot \frac{ \alpha_s - \alpha_j }{ \alpha_j - \alpha_\ell }} = 0. 
	\end{equation}
	Let 
	\[ M_i := M\inparen{\delta_i, \frac{\zeta_1}{\delta_i}, \frac{\zeta_2}{\delta_i}, \ldots, \frac{ \zeta_t }{\delta_i } } \]
	be as in Proposition~\ref{prop:repairscheme}.
	Then the matrix $M \in \F^{n \times rt}$ given by
	\begin{equation}\label{eq:defM}
	M = [M_1 | M_2 | \cdots | M_r ]
	\end{equation}
	is a multiple-repair matrix for $I$.
\end{lemma}
Notice that Lemma~\ref{thm:sufficient} does not make any claims about the bandwidth of this scheme; we will show below how to choose the $\delta_i$ so that \eqref{eq:needthis} holds, and so that the bandwidth is also small.
Because the columns of $M$ are columns of the $M_i$ and we have already established that these are dual codewords, the only thing left to prove is that Property~\ref{prop:fullrank} of Definition~\ref{def:megarepmat} holds; that is, that the $r \times rt$ matrix $M[I,:]$ has a trivial right kernel over $B$.  We will prove this in Section~\ref{sec:pfsuff}, but first, we will show how to use Lemma~\ref{thm:sufficient} in order to prove Theorem~\ref{thm:mainrs_small}.

\begin{proof}[Proof of Theorem~\ref{thm:mainrs_small}]
	Suppose without loss of generality that $I = \{\alpha_1,\ldots,\alpha_r\}$.
	Let $\zeta_1,\ldots, \zeta_t$ be any basis of $\F$ over $B$.
	We will choose the paramaters $\delta_1,\ldots, \delta_r$ successively
	so that Lemma~\ref{thm:sufficient} applies, and keep track of the bandwidth of the resulting repair matrix.
	
	Before we begin, we note that this approach---even without keeping track of the bandwidth---would immediately
	imply that $M$ is a multiple-repair matrix for $I$ with bandwidth at most $(n-r)\cdot r$; indeed, for all $i \in [n]\setminus I$,
	\[ \dim_B \inset{ \spn_B \inset{ \st{M[i,:]} } } \leq r, \]
	because $\st{M[i,:]} = \bigcup_{\ell=1}^r \st{M_\ell[i,:]}$,
	and for each $\ell$ we have
	\[ \dim_B \inset{ \spn_B \inset{ \st{M_\ell[i,:]}} } \leq 1. \]
	Then, Theorem~\ref{thm:main} implies that this gives a repair scheme for $I$ with bandwidth $b \leq (n-r)r$. 
	
	The approach above would recover the results of \cite{multfail} for $r=2,3$ and would generalize them to all $r$.
	However, in fact this calculation may be wasteful, and by choosing the $\delta_i$ carefully we can improve the bandwidth.
	More precisely, 
	we will try to choose $\delta_i$ so that the spans $\spn_B \inset{ \st{M[i,:]}}$ collide, and the dimension of the union is less than the sum of the dimensions.
	
	\begin{claim}\label{claim:FullRankTop}
		Let $\ell < r$ and suppose that $\delta_1,\ldots, \delta_{\ell-1}$ have been chosen.  Then for at least $|B|^{ t - (\ell-1) \cdot r + \ell(\ell-1)/2} - 1$ choices of $\delta \in \F^*$, 
		setting $\delta_{\ell} \gets \delta$ satisfies 
		\[ \tr_{\F/B} \inparen{  \frac{ \delta }{\delta_j } \cdot \frac{ \alpha_s - \alpha_j }{ \alpha_j - \alpha_\ell } } = 0 \]
		for all $j < \ell$ and all $s > j$.
	\end{claim}
	\begin{proof}
		For each $j < \ell$ and $s > j$, the above gives a linear requirement on $\delta_j$.  There are at most
		\[ \sum_{j=1}^{\ell-1} (r - j) = (\ell-1) \cdot r - \frac{ \ell(\ell - 1) }{2} \]
		such pairs $(j,s)$, so there are that many linear constraints.  Since $\F$ is a vector space over $B$ of dimension $t$, this proves the claim.
	\end{proof}

	We briefly recall some algebra.  For $\gamma \in \F^*$, the \em (multiplicative) coset \em $\gamma \cdot B^*$ is the set
	\[ \gamma \cdot B^* = \inset{ \gamma \cdot b \suchthat b \in B^* }. \]
	We say that $\gamma \equiv_{B^*} \gamma'$ if $\gamma \in \gamma'B^*$, and it is not hard to see that $\equiv_{B^*}$ is an equivalence relation that partitions $\F^*$ into $|\F^*|/|B^*|$ cosets of size $|B^*|$.  The group of all such cosets form the \em quotient group \em $\F^*/B^*$.
	The following observations follow directly from these definitions, as well as the definition of the matrix $M_\ell$.
	\begin{observation}\label{obs:goal}
		Suppose that $\st{M[i,:]} \subseteq \gamma_1 B^* \cup \gamma_2 B^* \cup \cdots \cup \gamma_c B^*$ for $c$ different cosets $\gamma_i B^*$.  
		Then 
		\[ \dim_B\inset{ \spn_B \inset{ \st{M[i,:]} } } \leq c. \]
	\end{observation}
	\begin{observation}\label{obs:collide}
		We have
		\[ \spn_B \inset{ \st{M_\ell[i,:]} } \subseteq \inparen{ \frac{ \delta_\ell }{ \alpha_i - \alpha_\ell }} B^* \cup \inset{0}. \]
	\end{observation}

Say that $(j,\ell)$ \textbf{collide} at $i$ if \[\frac{ \delta_j }{ \alpha_i - \alpha_j }B^* = \frac{ \delta_{\ell} }{ \alpha_i - \alpha_{\ell} }B^*.\]
Notice that $(j,\ell)$ collide at $i$ if and only if 
\[ \delta_{\ell } \in \inparen{\frac{\alpha_{i}-\alpha_{\ell}}{\alpha_i-\alpha_j}}\delta_j B^*.\]
Define
\[ S_{\ell} := \{i \in [n]\setminus I\,:\, \exists j<k\leq \ell \text{ s.t. } (j,k) \text{ collide at } i \}.\]

Consider choosing $\delta_1,...,\delta_r$ one at a time.  Choose $\delta_1 = 1$, and 
we will proceed by induction on $\ell$, assuming that we have chosen $\delta_1,\ldots, \delta_\ell$, with the inductive hypothesis that
\begin{equation}\label{eq:hyp}
|S_{\ell}| = \binom{\ell}{2} \cdot (|B|-1).
\end{equation}
Notice that the base case when $\ell=1$ is trivially satisfied with $\delta_1 = 1$, because $S_1 = \emptyset$.
Now suppose that the inductive hypothesis \eqref{eq:hyp} holds, and consider choosing $\delta_{\ell + 1}$.
For $j < \ell + 1$, define
\[ T_j^{\ell + 1} := \inset{ i \in [n] \,:\, (\ell+1,j) \text{ collide at } i }. \]
Notice that $T_j^{\ell+1}$ depends on the choice of $\delta_{\ell + 1}$, as the definition of colliding depends on $\delta_{\ell+1}$.

\begin{claim}\label{SmallSchemeClaim}
	Fix $j < \ell + 1$. For all but at most $r + \binom{\ell}{2}\cdot (|B|-1)$ cosets $\gamma B^*$, the following holds: 
	If $\delta_{\ell + 1} \in \gamma B^*$, then 	
	\begin{itemize}
		\item[(a)] $|T^{\ell+1}_j| = |B| - 1$, and 
		\item[(b)] $T^{\ell+1}_j \cap (I \cup S_\ell) = \emptyset.$
	\end{itemize}
\end{claim}
\begin{proof}
	Say $\gamma$ is \textbf{good} if 
	\[ \inparen{\frac{\alpha_i-\alpha_{\ell+1}}{\alpha_{i}-\alpha_{j}}}\delta_j \notin \gamma B^* \qquad \forall i \in I \cup S_{\ell}. \]
	That is, $\gamma$ is good if and only if choosing $\delta_{\ell+1} \in \gamma B^*$ would mean that $(\ell + 1, j)$ do \em not \em collide at any $i \in I \cup S_\ell$.
	We first claim that there are at most $|I \cup S_\ell|$ values of $\gamma \in F$ that are not good.  This follows from the fact that the map $h: z \mapsto \inparen{ \frac{ z - \alpha_{\ell+1} }{z - \alpha_j} } \delta_j $
	is a bijection from $F$ to $F$.
	Given this, and the fact that $|I \cup S_\ell| = |I| + |S_\ell| = r + {\ell \choose 2} \cdot (|B| -1)$, we conclude that there are at most $r + {\ell \choose 2} \cdot(|B| + 1)$ choices of $\gamma$ that are not good. 
	
	Suppose that $\gamma$ is good, and consider any choice of $\delta_{\ell + 1} \in \gamma B^*$.  
	Since $h$ as defined above is a bijection, there are $|B^*| = |B| - 1$ elements $\alpha_i \in F$ so that 
	\[ \inparen{ \frac{\alpha_i - \alpha_{\ell+1}}{\alpha_i - \alpha_j} } \delta_j \in \gamma B^* = \delta_{\ell + 1} B^*, \]
	and so $(\ell+1,j)$ collide at $i$, and by definition $T_j^{\ell+1}$ is this set.  This establishes (a).  Now (b) follows from the definition of good.
	This completes the proof of Claim~\ref{SmallSchemeClaim}.
\end{proof}

Claim~\ref{SmallSchemeClaim} is for a fixed $j$, and summing over all $j \leq \ell$, it implies that for all but at most $\ell \cdot (r + {\ell \choose 2}(|B|-1))\cdot(|B|-1)$ choices of $\delta_{\ell + 1}$, we have that for \em all \em  $j = 1,\ldots, \ell$, the set $T_j^{\ell+1}$ satisfies the conclusions (a) and (b) of Claim \ref{SmallSchemeClaim}.  We next claim that for such a choice of $\delta_{\ell+1}$, the sets $T_j^{\ell+1}$ are disjoint. 

\begin{claim}\label{claim:SmallSchemeDisjoint}
	Suppose that $\delta_{\ell+1}$ is chosen so that the conclusion (b) of Claim~\ref{SmallSchemeClaim} holds for all $j = 1,\ldots, \ell$.  
	Then for all $j < k \leq \ell$, $T^{\ell+1}_j \cap T^{\ell+1}_k = \emptyset.$
\end{claim}

\begin{proof}
	Suppose towards a contradiction that $i \in T_j^{\ell+1} \cap T_k^{\ell + 1}$.
	Then by the definition of $T_j^{\ell+1}$ and $T_k^{\ell+1}$, 
	$(\ell+1,j)$ and $(\ell+1,k)$ collide at $i$, which means that
	\[ \frac{\delta_j}{\alpha_i - \alpha_j} B^* = \frac{\delta_{\ell+1}}{\alpha_i - \alpha_{\ell + 1}} B^* = \frac{ \delta_k }{ \alpha_i - \alpha_k } B^*. \]
	This implies that $j$ and $k$ collide at $i$.  However, since $j, k \leq \ell$, this implies that $i \in S_\ell$. 
	But this contradicts the conclusion (b) of Claim~\ref{SmallSchemeClaim}, which states that $S_\ell \cap T^{\ell+1}_j = \emptyset$. 
\end{proof}

Now we finish the proof of Theorem~\ref{thm:mainrs_small}. 
Suppose that $\delta_{\ell+1}$ is chosen so that the conclusions (a) and (b) of Claim~\ref{SmallSchemeClaim} as well as the conditions of \eqref{eq:needthis} hold for $j = 1, \ldots, \ell$.
By definition we have
\[ S_{\ell + 1} = S_\ell \cup T_1^{\ell+1} \cup T_2^{\ell+1} \cup \cdots \cup T_\ell^{\ell+1}, \]
and by conclusion (b) of Claim~\ref{SmallSchemeClaim} and by Claim~\ref{claim:SmallSchemeDisjoint}, each of these sets of disjoint.
By induction, $|S_\ell| = {\ell \choose 2}(|B| - 1)$ and by the conclusion (a) of Claim~\ref{SmallSchemeClaim} we have $|T_j^{\ell+1}| = |B|-1$ for all $j \leq \ell$.  Thus, 
\begin{align*}
|S_{\ell+1}| &= |S_\ell| + \sum_{j=1}^\ell |T_j^{\ell+1}| \\
&= {\ell \choose 2}(|B| - 1) + \ell (|B| - 1) \\
&= { \ell + 1 \choose 2 } (|B| - 1).
\end{align*}
This establishes the inductive hypothesis for $\ell + 1$.  

Now, the requirement that
\[t > {r \choose 2} + \log_{|B|}\inparen{r\cdot\inparen{r+{r \choose 2}(|B|-1)}\cdot (|B|-1) + 1}\]
or equivalently
\[ |B|^{\left(t- {r \choose 2}\right)} - 1> r \cdot \inparen{ r + {r \choose 2}(|B| -1 ) } \cdot (|B| - 1) \]
in the statement of the Theorem~\ref{thm:mainrs_small}, in conjunction with Claim~\ref{claim:FullRankTop} 
implies that for all $\ell < r$, we can choose a $\delta_{\ell+1}$ so that the conclusions (a) and (b) of Claim~\ref{SmallSchemeClaim} as well as the conditions of \eqref{eq:needthis} hold. 
Now by choosing $\delta_1,\ldots,\delta_r$ in this way, 
we conclude by induction that
\[ |S_r| = {r \choose 2}(|B| - 1). \]
Since $S_r$ is the set of rows $i$ that have any collisions, each $i \in S_r$ contributes at most $r-1$ to the bandwidth, while $i \in [n] \setminus (I \cup S_r)$ may contribute $r$.  Thus, with this choice of $\delta_1,\ldots, \delta_r$, the bandwidth of the resulting scheme is at most
\[ {r \choose 2}(|B| - 1) (r-1) + \inparen{n - r - {r \choose 2}(|B|-1)}\cdot r = (n-r)\cdot r - {r \choose 2}(|B| - 1). \]
\end{proof}

\subsubsection{Proof of Lemma~\ref{thm:sufficient}}\label{sec:pfsuff}
In this Section, we prove Lemma~\ref{thm:sufficient}.  
Because we have made no claims about the bandwidth, we only need to show that the sub-matrix $M[I,:]$ has full rank, in the sense that for all nonzero $y \in B^{rt}$, $M[I,:]y \neq 0$.  
To save on notation, 
for the rest of this proof, let $A \in \F^{r \times rt}$ denote the matrix $M[I,:]$.

As in the proof of Theorem~\ref{thm:main}, it suffices to show that the $B$-linear map $\vphi_A: \F^r \to B^{rt}$ given by
\[ \vphi_A(x) = \tr_{\F/B}( x^T A ) \]
is invertible, where $\tr_{\F/B}$ is applied coordinate-wise.
Our proof basically follows from analyzing the $LU$-decomposition of this matrix $A$.  
That is, we will give an algorithm, consisting of row operations, which preserve the invertibility of $\vphi_M$ and after which $M$ will become a matrix of the form 
\begin{equation}\label{eq:form}
M' = \begin{bmatrix} \vec{\zeta} & * & * & \cdots & * \\ \vec{0} & \vec{\zeta} & * & \cdots & * \\
\vec{0} & \vec{0} & \vec{\zeta} & \cdots & * \\
& \vdots \\
\vec{0} & \vec{0} & \vec{0} & \cdots & \vec{\zeta} 
\end{bmatrix}, 
\end{equation}
where $\vec{0} \in \F^t$ denotes the vector of $t$ zeros and $\vec{\zeta} = (\zeta_1,\zeta_2, \ldots, \zeta_t)$.  Recall that $\zeta_1,\ldots, \zeta_t$ is the basis chosen in the statement of the lemma.
The matrix $M'$ as in \eqref{eq:form} clearly has the desired property ($\vphi_{M'}$ is invertible), and so this will finish the proof. 

\begin{remark}[Picking a basis]
	In what follows, we just go through the LU decomposition of the matrix $A$, and show that the assumption \eqref{eq:needthis} implies that the result has the form \eqref{eq:form}.  Unfortunately, this familiar argument may seem less familiar because we have not picked a basis for $\F$ over $B$, and are instead working with the trace functional.  If the reader prefers, they may imagine picking a basis for $\F$ over $B$, and working with square $rt \times rt$ matrices over $B$.  However, for this exposition we choose to leave things as they are to avoid the additional level of subscripts that picking a basis would require.
\end{remark}

Given $\gamma \in \F$, define the map $f_\gamma: \F \to \F$ by
\[ f_{\gamma}( x ) = \gamma \cdot \tr_{\F/B}\inparen{ \frac{ x}{\gamma} } . \]
Then extend $f_\gamma$ to $f_\gamma :\F^{rt} \to \F^{rt}$ by acting coordinate-wise.
\begin{observation}\label{obs:rowop}
	Let $A \in \F^{r \times rt}$, and
	suppose that $A'$ is obtained from $A$ by adding $f_\gamma(A[i,:])$ to $A[j,:]$, so 
	\[ A'[j,:] = A[j,:] + f_\gamma(A[i,:]) \]
	and for $\ell \neq j$, we have $A'[\ell,:] = A[\ell,:]$.
	Then $\inset{ \tr_{\F/B}(x^T A) \suchthat x \in \F^r } = \inset{ \tr_{\F/B}(x^T A') \suchthat x \in \F^r}$.
	In particular, $\vphi_A$ is invertible if and only if $\vphi_{A'}$ is invertible.
	(As before, above $\tr_{\F/B}$ is applied coordiatewise).
\end{observation}
\begin{proof}
	Given $x \in \F^r$, consider $x'$ given by
	$ (x')_i = x_i +  \tr( \gamma x_j )/\gamma. $
	Then $\tr_{\F/B}(x^T A') = \tr_{\F/B}((x')^T A).$
\end{proof}

Now consider the following algorithm:

\begin{itemize}
	\item[] Algorithm \textbf{LU}(A):
	\begin{itemize}
		\item[] $A^{(0)} \gets A$
		\item[] \textbf{for} $j = 1, \ldots, r-1$:
		\begin{itemize}
			\item[] \textbf{for} $s = j+1, \ldots, r$:
			\begin{itemize}
				\item[] $A^{(j)}[s,:] \gets A^{(j-1)}[s,:] + f_{\delta_j/(\alpha_s - \alpha_j)}( A^{(j-1)}[j,:] )$
			\end{itemize}
			\item[] \textbf{end for}
		\end{itemize}
		\item[] \textbf{end for}
		\item[] \textbf{return} $A^{(r)}$.
	\end{itemize}
\end{itemize}

Because of Observation~\ref{obs:rowop}, we see that $\vphi_{A^{(r)}}$ is invertible if and only if $\vphi_A$ is invertible.
Moreover, we claim that, if \eqref{eq:needthis} is met, then $A^{(r)}$ has the form \eqref{eq:form}.  To see this we proceed by induction.
Write $A = [A_1 |A_2 | \cdots | A_r ]$, where $A_j \in \F^{r \times t}$,  and similarly write
\[ A^{(j)} = [ A_1^{(j)} | A_2^{(j)} | \cdots | A_r^{(j)} ]. \]
Then the inductive hypothesis is that for all $i \leq j$, $A_i^{(j)}$ is equal to $A_i$ on the first $i$ rows and is equal to zero otherwise.  That is, the first $j$ blocks of $A^{(j)}$ have zeros in the form of \eqref{eq:form}, and all nonzero entries are the same as in $A$.
The base case is immediate for $j = 0$, with the notational assumption that any statement about $M[\ell, c]$ for $c \leq 0$ is vacuously true.

Now assuming that this holds for $j-1$, we establish it for $j$.  
First notice that, because of the inductive hypothesis, the first $j-1$ blocks do not change.
For block $j$, and a row $s > j$, we update
\[ A^{(j)}_j[s,:] \gets A_j^{(j-1)}[s,:] + f_{\delta_j/(\alpha_s - \alpha_j)} ( A^{(j-1)}_j[j,:] )
= A_j[s,:] + f_{\delta_j/(\alpha_s - \alpha_j)} (A_j[j,:]) \]
again using the inductive hypothesis.
By construction, for all $w \in [t]$,
\[ A_j[s,w] = \frac{ \delta_j}{\alpha_s - \alpha_j} \cdot \tr_{\F/B} ( \zeta_w \cdot( \alpha_s - \alpha_j) / \delta_j). \]
Then the update is
\begin{align*}
f_{\delta_j/(\alpha_s - \alpha_j)}( A_j[j,w] ) &= \frac{ \delta_j }{\alpha_s - \alpha_j } \cdot \tr_{\F/B}\inparen{ \frac{ \alpha_s - \alpha_j }{\delta_j } \cdot A_j[j,w] } \\
&= \frac{ \delta_j }{ \alpha_s - \alpha_j } \cdot \tr_{\F/B} ( \zeta_w \cdot( \alpha_s - \alpha_j) / \delta_j) \\
&= A_j[s,w],
\end{align*}
and so $A_j[s,w] - f_{\delta_j/(\alpha_s - \alpha_j)}(A_j[j,w]) = 0$ for all $s > j$.
Thus, this operation indeed zeros out all but the first $j$ rows of $A_j$.  

Next, consider a block $\ell > j$ and a row $s > j$.  We need to show that $A_\ell^{(j-1)}[s, w]$ does not change in the $j$'th iteration, namely that the update is zero.  Now we have, by induction and by definition, respectively, that
\[ A_\ell^{(j-1)}[j,w] = A_\ell[j,w] = \frac{ \delta_\ell}{\alpha_j - \alpha_\ell} \cdot \tr_{\F/B} ( \zeta_w \cdot( \alpha_j - \alpha_\ell) / \delta_\ell). \] 

A computation similar to the one above establishes that
\begin{align*}
f_{\delta_j/(\alpha_s - \alpha_j)}( A_\ell[j,w] ) 
&= \frac{ \delta_j }{ \alpha_s - \alpha_j } \cdot \tr_{\F/B} ( \zeta_w \cdot( \alpha_j - \alpha_\ell )/\delta_\ell ) \cdot \tr_{\F/B}\inparen{
	\frac{ \delta_\ell }{\delta_j} \cdot \frac{ \alpha_s - \alpha_j }{\alpha_j - \alpha_\ell } }
\\
&= 0,
\end{align*}
where in the final line we have used \eqref{eq:needthis}. 
This establishes the other part of the inductive hypothesis (that all other entries remain the same) and this completes the proof.

\section{Conclusion}
In this work, we have extended the framework of~\cite{gw16} to handle multiple failures, and instantiated this framework in two ways to give improved results for Reed-Solomon codes with multiple failures.  However, several open problems remain.  We highlight a few promising directions below.
\begin{enumerate}
\item We have no reason to believe that the bound in Theorem~\ref{thm:mainrs_large} is asymptotically tight. We leave it as an open question to either obtain an improved construction or a matching lower bound.  As a concrete question, we may ask if it is possible to repair full-length Reed-Solomon codes from $r = \omega(n^R)$ failures with nontrivial bandwidth.
\item Our work is restricted to the centralized model for repair of multiple nodes.  
On the other hand, the work of \cite{multfail} obtains results for Reed-Solomon codes for $r=2,3$ in other models where the communication between the nodes is taken into account when measuring the bandwidth; our framework does not apply there.  Could our techniques be adapted to apply to this model as well?
\item Finally, we have instantiated our framework for full-length RS codes, but it may be interesting in other parameter regimes as well.  Recently, Ye and Barg have given a construction of a Reed-Solomon code which achieves the cut-set bound for multiple failures~\cite{YB17}; however, the subpacketization is (necessarily) extremely large.  Can we instantiate our framework to obtain RS codes (or other scalar MDS codes) with repair bandwidth that improves over our scheme, but which still have small (sublinear in $n$) subpacketization?
\end{enumerate}

\bibliographystyle{plain}
\bibliography{refs}

\end{document}